\pdfoutput=1
\documentclass[sigconf]{acmart}
\usepackage{xspace}
\usepackage{ifthen}
\usepackage{multirow}
\usepackage{wrapfig}
\usepackage{multicol}
\providecommand{\Description}[2][]{}

\makeatletter
\@ifclassloaded{acmart}{%
  \setlength{\emergencystretch}{2em}%
}{%
  \setlength{\abovecaptionskip}{0pt}%
  \setlength{\belowcaptionskip}{0pt}%
  \g@addto@macro\normalsize{%
    \setlength{\abovedisplayskip}{0pt}%
    \setlength{\abovedisplayshortskip}{0pt}%
    \setlength{\belowdisplayskip}{0pt}%
    \setlength{\belowdisplayshortskip}{0pt}%
  }%
  \setlength{\textfloatsep}{2pt}
  \setlength{\intextsep}{2pt}
}
\makeatother

\usepackage{algorithm}
\usepackage{algpseudocode}

\usepackage{enumitem}

\usepackage{subcaption}

\makeatletter
\@ifclassloaded{acmart}{}{\usepackage[table]{xcolor}}
\makeatother

\newcommand{\problem}{structure-aware table retrieval\xspace}

\newcommand{\PROBLEM}{Structure-aware Table Retrieval\xspace}

\newcommand{\method}{\textsc{GRAFT}\xspace}

\newcommand{\metric}{IGMS\xspace}
\usepackage{listings}

\newcommand{\accu}{\textit{Accuracy}\xspace}

\newcommand{\rmse}{\textit{RMSE}\xspace}

\newcommand{\spider}{\textit{Spider}\xspace}
\newcommand{\bird}{\textit{BIRD}\xspace}
\newcommand{\open}{\textit{Open Data}\xspace}

\usepackage{amsmath}
\usepackage{amsthm}
\makeatletter
\@ifclassloaded{acmart}{}{\usepackage{amssymb}}  
\makeatother

\DeclareMathOperator*{\argmax}{arg\,max}

\newtheorem{theorem}{Theorem}
\newtheorem{proposition}[theorem]{Proposition}

\theoremstyle{definition}

\theoremstyle{remark}

\AtBeginDocument{%
  }

\begin{document}

\title{GRAFT: Graph-Matched Retrieval and Fusion of Tables in Data Lakes}

\author{Daomin Ji}
\email{daomin.ji@student.rmit.edu.au}
\affiliation{%
  \institution{RMIT University}
  \city{Melbourne}
  \country{Australia}}

\author{Hui Luo}
\email{huil@uow.edu.au}
\affiliation{%
  \institution{University of Wollongong}
  \city{Wollongong}
  \country{Australia}}

\author{Zhifeng Bao}
\email{zhifeng.bao@uq.edu.au}
\affiliation{%
  \institution{The University of Queensland}
  \city{Brisbane}
  \country{Australia}}

\author{Shane Culpepper}
\email{s.culpepper@uq.edu.au}
\affiliation{%
  \institution{The University of Queensland}
  \city{Brisbane}
  \country{Australia}}

\author{Shazia Sadiq}
\email{shazia@eecs.uq.edu.au}
\affiliation{%
  \institution{The University of Queensland}
  \city{Brisbane}
  \country{Australia}}

\renewcommand{\shortauthors}{Ji et al.}

\begin{abstract}
Autonomous data agents resolve analytical queries by retrieving and reasoning over evidence in tabular data lakes. Existing methods score tables independently against the query and ignore the joinability and unionability that link them, returning fragmented evidence that downstream agents cannot integrate. We propose \method (\textbf{G}raph-matched \textbf{R}etrieval \textbf{a}nd \textbf{F}usion of \textbf{T}ables), structured around two principal contributions. First, we cast table retrieval as a graph matching problem between a query-derived intent graph and a heterogeneous data lake graph, and introduce \metric, a log-determinant reward that couples semantic relevance, structural compatibility, and evidence diversity in a single objective. Second, we recast subgraph generation as a Markov decision process and learn a value function via implicit Q-learning on self-generated trajectories produced by a canonical compression operator that inverts the homomorphism. We further design a three-stage online pipeline that exploits anchor reachability, predicate admissibility, and reward monotonicity to greatly prune the candidate space before exact \metric evaluation. On Spider and BIRD adapted to the tabular data lake setting, \method achieves the best Recall, Precision, F1, and Sufficiency among point-wise, greedy-expansion, and structure-aware baselines, with relative gains of $7.8\%$ in F1 and $10.6\%$ in Sufficiency over the strongest baseline, while maintaining high search efficiency.
\end{abstract}


\keywords{table retrieval, data lakes, data discovery, graph matching, offline reinforcement learning}

\maketitle

\section{Introduction}\label{sec:intro}

Autonomous data agents~\citep{zhu2025survey,tang2025llm_agent_as_data_analyst} powered by large language models (LLMs) have shown strong proficiency in automating data-centric tasks such as question answering~\citep{lewis2020rag,karpukhin2020dpr,yang2018hotpotqa,Tang2025STRaptor,Zhang2025AixelAsk}, fact checking~\citep{thorne2018fever}, code generation~\citep{chen2021codex,yang2024sweagent}, and mathematical reasoning~\citep{cobbe2021gsm8k,lewkowycz2022minerva}. To resolve a natural language query, these agents first retrieve relevant information from heterogeneous sources and then reason over it~\citep{wang2023llm_agents_survey,xi2023rise_agents,yao2023react,karpas2022mrkl}. Consequently, the quality of downstream reasoning hinges directly on the relevance and completeness of the retrieved content~\citep{lewis2020rag,karpukhin2020dpr}.

While retrieval over unstructured data such as text and images has been extensively studied~\citep{karpukhin2020dpr,khattab2020colbert,radford2021clip,johnson2019billion}, retrieval over tabular data lakes~\citep{wang2023solo,guo2025birdie}, i.e., large repositories of structured tables aggregated from heterogeneous sources, remains comparatively underdeveloped. 
Retrieval over tabular data lakes exhibits two defining characteristics that distinguish it from retrieval over text or image corpora. First, tables within a data lake are typically fragmented across heterogeneous sources, and independent retrieval over them yields disconnected candidates that cannot be integrated into a single answer. Resolving a query therefore requires identifying multiple tables linked through joinability and unionability and synthesizing them into a coherent integrated table, rather than returning isolated candidates. For example, the query in Fig.~\ref{fig: beneficial examples} requires joining the instructor records (in $T_2$, $T_3$) to their department names (in $T_5$) through the (instructor, department) bridge table $T_4$. Without $T_4$, no other combination of retrieved tables can reconstruct the join path needed to filter on the Computer Science department, regardless of how many candidate tables are returned. Second, queries over tabular data are predominantly analytical and frequently involve aggregation operations such as SUM, COUNT, and AVG, whose correctness depends on the completeness of the supporting evidence rather than on the relevance of any single table. The retrieved set must therefore exhibit \emph{diversity}, contributing complementary attributes and instances that jointly cover the analytical population while suppressing redundant evidence units that project to the same underlying content. For example, $T_2$ and $T_3$ in Fig.~\ref{fig: beneficial examples} are unionable tables that record instructors of different roles (Professors in $T_2$; Lecturers and Assistant Professors in $T_3$) at the same set of universities. Retrieving only one of the two undercounts the Computer Science instructors at every university and distorts the COUNT aggregation. Both tables are required, and a retriever that treats them as near-duplicate evidence systematically biases the analytical answer.



Thus, we study the problem \textit{\PROBLEM}: given a natural language query and a tabular data lake, the goal is to identify a set of interconnected tables, structurally linked through joinability and unionability, that collectively satisfy the user's intent. The \textit{\problem} task raises two core challenges:
\textbf{(1) Unified utility design.} As discussed above, retrieval over tabular data lakes introduces desiderata beyond surface relevance, requiring a utility that couples three signals in a single objective: \emph{semantic relevance} between the query and individual evidence units, \emph{structural compatibility} between the relations implied by the query and executable join or union paths in the lake, and \emph{evidence diversity} that suppresses redundancy among selected units.
\textbf{(2) Combinatorial search space.} The search space grows exponentially in the number of tables. In a data lake of $M$ tables with average degree $N$ in the graph induced by joinability and unionability, retrieving a collection of $L$ interconnected tables admits on the order of $\mathcal{O}(M \cdot N^{L-1})$ candidates, which precludes exhaustive enumeration.

\begin{figure*}[t]
    \centering
    \includegraphics[width=\linewidth]{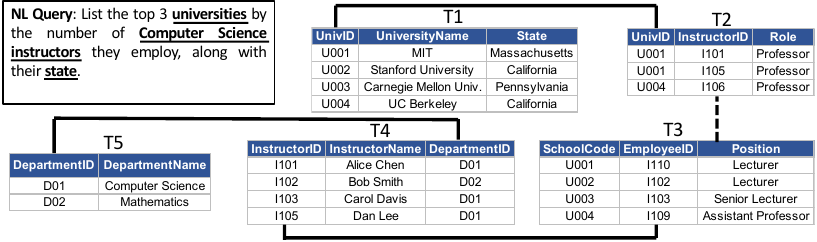}
       \caption{A motivating example of composition-guided table retrieval, where dashed lines denote unionability and solid lines denote joinability. The retrieved tables can jointly answer the query via $T_1 \bowtie (T_2 \cup T_3) \bowtie T_4 \bowtie T_5$. Note that schemas (e.g., column names) of tables in the data lake may not be aligned.}
    \Description{Five tables in a data lake connected by join and union edges, which together answer a natural-language query about Computer Science instructors.}
    \label{fig: beneficial examples}
\end{figure*}

Existing methods fall short of jointly meeting these requirements. \emph{Point-wise retrieval methods}~\citep{wang-etal-2022-table,guo2025birdie} scores each table against the query and returns a top-$k$ list, yielding candidates that are individually relevant yet structurally disconnected and frequently redundant. \emph{Greedy expansion methods}~\citep{wu2025mmqa} extends the selection by appending the table most relevant to either the query or the most recently selected table, but relevance to either anchor does not assure structural compatibility. Recently, JAR~\citep{chen-etal-2024-table} makes the first attempt to retrieve structurally compatible tables by reformulating retrieval as a mixed-integer program (MIP) over the joinable tables. However, JAR adopts a simple linear combination of relevance and joinability that neglects diversity aspect, and the NP-hardness of the MIP formulation causes the approach to scale poorly with lake size.

To address these challenges, we propose \textbf{\method} (\textbf{G}raph-matched \textbf{R}etrieval \textbf{a}nd \textbf{F}usion of \textbf{T}ables), which formulates \problem as a graph-matching problem between two typed graphs: an intent graph $\mathcal{G}_I$ distilled from the natural-language query, encoding the entities, attributes, and structural relations the user wants integrated, and a data lake graph $\mathcal{G}_\mathcal{D}$ encoding schema containment together with joinability and unionability across tables. Then, the problem reduces to extracting a connected subgraph of $\mathcal{G}_\mathcal{D}$ that admits a typed homomorphism from $\mathcal{G}_I$. To address (C1), inspired by mutual-information objectives in Gaussian-process information gain~\citep{krause2008nearoptimal}, we propose the \emph{Information-theoretic Graph Matching Score} (\metric). To address (C2), we cast subgraph generation as a Markov decision process, learn a value function via implicit Q-learning~\citep{kostrikov2022iql} on self-generated trajectories, and apply three pruning stages that exploit anchor reachability, predicate admissibility, and reward monotonicity, reducing the candidate set from $\mathcal{O}(M \cdot N^{L-1})$ to a shortlist of size $\mathcal{O}(BL)$ on which exact \metric is evaluated.

In summary, our contributions are as follows.

\begin{itemize}[leftmargin=*,noitemsep,topsep=0pt]
\item We recast \problem as graph-matching between a query-derived intent graph and a data-lake graph, and introduce \metric, a log-determinant reward that couples semantic relevance, structural compatibility, and evidence diversity in a single objective .
\item  We learn a value function over the MDP via offline implicit Q-learning~\citep{kostrikov2022iql} on self-generated trajectories, and design a three-stage pruning pipeline to achieve search efficiency.
\item We empirically validate \method on Spider and BIRD adapted to the tabular data lake setting, where it outperforms point-wise, greedy-expansion, and structure-aware baselines on every retrieval metric, with relative gains of $7.8\%$ in F1 and $10.6\%$ in Sufficiency over the strongest baseline, at latency on par with greedy-expansion methods.
\end{itemize}
\section{Background}\label{sec:background}

Table retrieval has been studied for decades across machine learning, databases, and information retrieval. At an abstract level, given a tabular corpus $\mathcal{D}=\{T_1,\dots,T_M\}$ and a query $Q$, a retrieval system returns a subset $\mathcal{D}'\subseteq\mathcal{D}$ of size at most $L$ that maximizes a query-conditional relevance score,
\begin{equation}
\mathcal{D}^{\star} \;=\; \arg\max_{\mathcal{D}'\subseteq\mathcal{D},\,|\mathcal{D}'|\le L} \sum_{T \in \mathcal{D}'} r(T\mid Q).
\label{eq:generic-retrieval}
\end{equation}

The formulation of the query $Q$ and the relevance scoring function $r$ inherently depend on the modality used to express the user's information need. Traditional retrieval paradigms have primarily relied on modalities such as keywords~\citep{fernandez2018aurum,zhu2016lshensemble,brickley2019google}, base tables~\citep{zhu2019josie,dong2023deepjoin,santos2021correlation,nargesian2018table,fan2023semantics,khatiwada2023santos,bogatu2020dataset}, and images~\citep{ji2024navigating,ji2025dataset}. More recently, however, the proliferation of large language models has established natural language (NL) as the dominant query modality. Unlike structural or artifact-based inputs, NL empowers users to flexibly articulate complex, multi-faceted information needs without necessitating prior knowledge of the underlying schema or the provision of a seed artifact.

In this paper, we mainly focus on natural-language (NL)-driven retrieval over \emph{tabular data lakes}, e.g.,large repositories of heterogeneous tables harvested from sources such as open data portals, enterprise warehouses, and web crawls. Unlike relational databases, tabular data lakes lack a global schema and rarely provide explicit inter-table relationships. We organize prior work along two retrieval strategies, characterized by how the output $\mathcal{D}'$ is constructed from $\mathcal{D}$, and then contrast both with the formulation we adopt for \PROBLEM.

\paragraph{Point-wise Retrieval.} The dominant strategy scores each table independently against the query and returns the top-$L$ candidates,
\begin{equation}
\mathcal{D}' \;=\; \mathrm{Top}_{L}\bigl\{\, r(T\mid s) \,:\, T \in \mathcal{D} \,\bigr\}
\label{eq:pointwise}
\end{equation}
The score $r$ takes different forms across methods. Learned table-text encoders such as TaPas~\citep{herzig2020tapas}, TaBERT~\citep{yin2020tabert}, and Solo~\citep{wang2023solo} embed tables and queries into a shared latent space and instantiate $r(T \mid s) = \langle \phi(T),\, \psi(s) \rangle$, where $\phi$ and $\psi$ are learned table and query encoders respectively. LLM-judged hybrid pipelines~\citep{balaka2025pneuma} instead aggregate lexical, semantic, and reasoning signals through a prompted scorer $r(T \mid s) = g\bigl(r_{\mathrm{lex}},\, r_{\mathrm{sem}},\, r_{\mathrm{llm}}\bigr)$, where $g$ is an LLM-induced fusion function. All such point-wise methods share a structural limitation: each table is scored in isolation, so when the answer is distributed across several tables no individual $r(T \mid s)$ is large, and the relational operations required to integrate the retrieved tables are not produced as part of the output.

\paragraph{Greedy Expansion.} More recent methods~\citep{wu2025mmqa} acknowledge that a query may require several tables to answer jointly, in the manner of multi-hop reasoning. They therefore adopt an iterative process that, at each step $t$, appends the table most relevant to the query or to the most recently selected table:
\begin{equation}~\label{eq:greedy}
T_{t+1} \;=\; \arg\max_{T \in \mathcal{D} \setminus \mathcal{S}_t}\, r\bigl(T \mid s, T_t\bigr), \qquad \mathcal{S}_{t+1} = \mathcal{S}_t \cup \{T_{t+1}\}
\end{equation}
where $\mathcal{S}_t$ is the selection after $t$ steps and $T_t$ is the most recently appended table. However, relevance to either anchor does not imply that the selected tables can be joined or unioned. The retrieved sequence therefore still suffers from structurally disconnected candidates that cannot be integrated into a coherent answer in practice.

\paragraph{\PROBLEM.} Beyond point-wise retrieval (Eq.~\ref{eq:pointwise}) and greedy iterative expansion (Eq.~\ref{eq:greedy}), the recent JAR~\citep{chen-etal-2024-table} makes the first attempt to retrieve structurally compatible tables by reformulating the task as a mixed-integer program (MIP) over joinable tables; however, its linear combination of relevance and joinability overlooks diversity, and the NP-hardness of the MIP formulation causes the approach to scale poorly with lake size. In contrast, we aim to retrieve a set of tables interconnected through join and union operations that jointly fulfill the user's information need expressed in $Q$. We formulate this as
\begin{equation}\label{eq:our-formulation}
\bigl({\mathcal{D}'}^{\star},\, \Phi^{\star}\bigr) \;=\; \argmax_{\substack{\mathcal{D}'\subseteq\mathcal{D},\,|\mathcal{D}'|\le L,\, \Phi \in \mathcal{F}(\mathcal{D}')}}\; U\bigl(Q,\, \mathcal{D}',\, \Phi\bigr)
\end{equation}
where $\mathcal{F}(\mathcal{D}')$ denotes the space of valid join-union operator sequences that integrate $\mathcal{D}'$ into a single, coherent, and \emph{non-empty} table.
\section{Information-Theoretic Graph Matching Metric}
\label{sec:iges}

\subsection{Table Retrieval As a Graph Matching Problem}
\label{sec:graphs-problem}
 
\noindent\textbf{Data-lake graph.} We model the data lake as a heterogeneous graph $\mathcal{G}_\mathcal{D}=(\mathcal{V}, \mathcal{E})$ comprising three node types and two edge types. The node types are \emph{tables}, \emph{columns}, and \emph{values} (one node per cell). \emph{Containment} edges link each table to its columns and each column to its value nodes, while \emph{structural} edges encode \textsc{join} and \textsc{union} relations between tables and are derived using existing joinability and unionability metrics~\citep{khatiwada2023santos,dong2023deepjoin}.

\noindent\textbf{Intent graph.}
Different from existing table retrieval methods that decompose queries into a flat set of atoms~\citep{chen-etal-2024-table,wu2025mmqa}, we model the user's intent as a typed graph $\mathcal{G}_I$ structured as a forest of \emph{intent trees} connected by cross-tree edges. Each intent tree corresponds to one coherent data need and follows a three-level hierarchy: the root specifies the target \textsc{entity}, its children enumerate the required \textsc{attributes}, and each attribute may carry a \textsc{value} leaf. The graph encodes two types of structural relationship implicitly. Within a single intent tree, the underlying data need may be vertically fragmented across multiple tables in the tabular data lake, which must be combined by \textsc{union}. Across distinct intent trees, the corresponding data needs must be combined by \textsc{join}, often through bridge tables that the user does not name in the query but that the data lake must supply for the join to execute. For example, the query in Fig.~\ref{fig: beneficial examples} induces the two-tree intent graph shown in Fig.~\ref{fig:intent_graph}: a \textsc{University} tree with attributes Name and State, an \textsc{Instructor} tree whose Department attribute carries the value leaf Computer Science, and a cross-tree edge encoding the employment relation between the two entities.

\begin{figure}[t]
\centering
\includegraphics[width=0.9\linewidth]{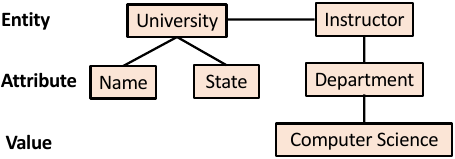}
\caption{Intent graph $\mathcal{G}_I$ for the running query in Fig.~\ref{fig: beneficial examples}: a forest of two intent trees connected by a cross-tree edge.}
\Description{An intent graph consisting of a University tree and an Instructor tree connected by a cross-tree employment edge.}
\label{fig:intent_graph}
\end{figure}

We obtain $\mathcal{G}_I$ from the natural-language query $s$ by prompting an LLM to emit the intent graph as a JSON object conforming to a fixed schema. The prompt specifies the schema (typed nodes, edges, and edge labels) and is organized into four blocks: a task description, the JSON output schema, a small set of constraints, and three in-context examples chosen to exercise (i) a single intent tree with a value constraint, (ii) a two-tree extraction with a cross-tree join, and (iii) a three-tree extraction in which a bridge entity unmentioned in the query is required to connect two named entities. The full prompt, schema, and demonstrations are provided in Appendix~\ref{app:prompt}.

\noindent\textbf{Problem reformulation.} The intent graph $\mathcal{G}_I$ specifies, at the schema level, what any valid integrated table must witness: each node fixes a type-tier element that the result must contain, and each edge fixes a structural dependency that the result must realize. Each edge type $r$ in $\mathcal{G}_I$ is paired with a fixed compatibility predicate $\chi_r$ that enumerates the typed paths in $\mathcal{G}_\mathcal{D}$ realizing $r$: column-of-table containment paths for intra-tree hierarchical edges, and joinable structural paths for cross-tree \textsc{join} edges. \textsc{Union} is not encoded as a separate edge type. Instead, a single intent tree may be witnessed by a union of unionable lake tables, with the union resolved at the homomorphism level rather than through an explicit edge predicate. A $\Phi$ is therefore valid precisely when its induced subgraph $H_\Phi \subseteq \mathcal{G}_\mathcal{D}$ admits a homomorphism from $\mathcal{G}_I$ in which every intent edge maps to a $\chi_r$-compatible path and every intent tree maps to a connected, possibly unioned, lake-table cluster. Writing $\mathcal{A}_I = V_I \cup E_I$ for the set of \emph{intent atoms} ($a := |\mathcal{A}_I|$), we cast retrieval as the constrained matching
\begin{equation}\label{eq:objective-matching}
\Phi^\star \;=\; \argmax_{\Phi}\;\; \mathcal{U}\bigl(H_\Phi;\, \mathcal{G}_I\bigr)
\quad \text{s.t.} \quad |V_T(H_\Phi)| \leq L.
\end{equation}
The matching is a quotient rather than a subgraph isomorphism: a single intent edge may abstract a multi-hop path in $\mathcal{G}_\mathcal{D}$ (e.g., the cross-tree edge in Fig.~\ref{fig:intent_graph} is realized by $T_1 \bowtie (T_2 \cup T_3)$ in Fig.~\ref{fig: beneficial examples}), and a single intent tree may be witnessed by a union of lake tables (e.g., the \textsc{Instructor} tree by $T_2 \cup T_3$). Multiple physical units may therefore collectively witness one intent atom.
\subsection{From Desiderata to the Unified \metric Metric}
\label{sec:metric}

As motivated in Sec.~\ref{sec:intro}, the metric must score $H_\Phi$ along three coupled desiderata, organized from per-atom support to set-level redundancy.
(1) \textit{Semantic relevance.} Each intent atom is supported by an evidence unit whose semantics match it, not merely its type.
(2) \textit{Structural compatibility.} Each intent edge is realized by a typed path in $H_\Phi$ satisfying $\chi_r$, not by two co-relevant tables.
(3) \textit{Diversity.} Distinct evidence units contribute non-overlapping column support, avoiding redundant projections that would double-count in $T_\Phi$.

\noindent\textbf{From desiderata to objective.} We translate the three desiderata into a single objective in which each component corresponds directly to one of them. We summarize semantic relevance and structural compatibility in an evidence-by-atom \emph{support matrix} $M_\Phi \in \mathbb{R}^{q \times a}$, where $q$ is the number of evidence units extracted from $H_\Phi$ (graph-level objects: tables, columns, value summaries, paths, union components) and $a = |\mathcal{A}_I|$. For an intent node $u$, the entry is $M_\Phi[i, u] = \max_{x \in \mathcal{N}(z_i)} \psi(u, x)$, where the relevance score $\psi(u, x) = \mathbb{1}[\tau_I(u) \sim \tau_\mathcal{D}(x)] \cdot \sigma\!\bigl(\beta\, \cos(\mathrm{emb}(u), \mathrm{emb}(x))\bigr)$ fuses temperature-scaled cosine similarity with a fixed type-compatibility gate. For an intent edge $e = (u, r, v)$, the entry is $M_\Phi[i, e] = \max_{p} \psi_{\text{path}}(e, p)$, where the maximum is taken over $\chi_r$-compatible paths $p \in \mathcal{P}^{K, h}_{u,v}$ associated with $z_i$ and $\psi_{\text{path}}(e, p) = \psi(u, x_0) \cdot \bigl(\prod_{g \in p} \omega_\mathcal{D}(g)\bigr)^{1/|p|} \cdot \psi(v, x_{|p|})$ combines the endpoint match scores with the geometric mean of edge-confidence weights along the path. Diversity is encoded in an evidence-by-evidence covariance matrix $C_\Phi \in \mathbb{R}^{q \times q}$ that captures redundancy among evidence units via column-set Jaccard similarity:
$
C_\Phi[i, j] \;=\; \mathrm{ColJac}(z_i, z_j) + \sigma_0^2 \,\mathbb{1}[i = j],
$
with $\sigma_0^2 > 0$ a small jitter constant ensuring $C_\Phi \succ 0$. The support matrix and covariance are coupled into a single reward through the log-determinant of the resulting information matrix:
\begin{equation}
\mathcal{U}_{\text{\metric}}(Q, \Phi) \;=\; \frac{1}{a}\, \log\det\!\left(I_a + M_\Phi^\top C_\Phi^{-1} M_\Phi\right)
\label{eq:iges}
\end{equation}

\noindent\textbf{Information-theoretic interpretation.} Eq.~\eqref{eq:iges} is, up to a multiplicative constant, the mutual information between a Gaussian latent representation of the intent and a noisy linear observation through the retrieved evidence. Place a prior $\theta \sim \mathcal{N}(0, I_a)$ over an intent-atom support latent, and treat each evidence unit as a noisy observation $y_i = M_\Phi[i, \cdot]\,\theta + \epsilon_i$ with $\epsilon \sim \mathcal{N}(0, C_\Phi)$. Standard Gaussian-channel calculus gives $I(\theta; y) = \tfrac{1}{2}\log\det(I_a + M_\Phi^\top C_\Phi^{-1} M_\Phi)$, which equals $\mathcal{U}_{\text{\metric}}$ up to constant rescaling. The same expression is the objective of Bayesian D-optimal experimental design and Gaussian-process information gain~\citep{krause2008nearoptimal}, where it selects maximally-informative observation sets under correlated noise. Under this reading, $M_\Phi$ specifies how strongly each evidence unit observes each intent atom (semantic relevance and structural compatibility), $C_\Phi$ is the noise covariance whose inversion down-weights observations whose noise is correlated (diversity), and the log-determinant is the entropy reduction that the retrieved evidence achieves about the intent. \metric is therefore not a heuristic combination of three signals but a principled information-gain metric.

The \metric reward in Eq.~\eqref{eq:iges} is monotone, submodular, and redundancy-suppressing in the set of evidence units, which yields a $(1-1/e)$-approximation guarantee for greedy selection under the cardinality constraint $|V_T(H_\Phi)| \leq L$~\citep{nemhauser1978submodular} and is the formal basis of the value-guided beam search in Sec.~\ref{sec:online}. Formal statements and proofs are given in Sec.~\ref{sec:proofs}.

\subsection{Properties of the \metric Metric}
\label{sec:proofs}

This subsection formalizes and proves the three properties of the \metric reward claimed in Sec.~\ref{sec:metric}: monotonicity, submodularity, and redundancy suppression. Throughout, we write $J_\Phi = I_a + M_\Phi^\top C_\Phi^{-1} M_\Phi$ for the information matrix on the evidence set of $\Phi$, with $M_\Phi$ and $C_\Phi$ as defined in Sec.~\ref{sec:metric}. Adding an evidence unit $z$ to $\Phi$ extends $M_\Phi$ by appending a row $m^\top \in \mathbb{R}^{1 \times a}$ and extends $C_\Phi$ by appending a row and column $(c,\, \sigma^2)$, where $c \in \mathbb{R}^q$ records the redundancy covariance of $z$ with existing units and $\sigma^2 = 1 + \sigma_0^2$ is its self-covariance. We use $\Delta(z \mid \Phi) := \mathcal{U}_{\metric}(Q,\, \Phi \cup \{z\}) - \mathcal{U}_{\metric}(Q,\, \Phi)$ for the marginal gain.

\begin{proposition}[Monotonicity]
\label{prop:monotone}
For any synthesis $\Phi$ and any evidence unit $z \notin \Phi$, $\Delta(z \mid \Phi) \geq 0$.
\end{proposition}

\begin{proof}
By the matrix-determinant lemma applied to the rank-one update of $J_\Phi$ induced by the new row $m$,
\[
\log\det J_{\Phi \cup \{z\}} - \log\det J_\Phi \;=\; \log\bigl(1 + \rho^{*-1}\, m^\top S\, m\bigr),
\]
where $\rho^{*} = \sigma^2 - c^\top C_\Phi^{-1} c \geq \sigma_0^2 > 0$ is the Schur complement of $z$ in the augmented covariance, and $S \succeq 0$ is determined by $J_\Phi$ and the projection of $m$ onto the column span of existing units. Both $\rho^{*}$ and $m^\top S\, m$ are non-negative, so $\log(1 + \rho^{*-1}\, m^\top S\, m) \geq 0$. Dividing by $a > 0$ gives $\Delta(z \mid \Phi) \geq 0$.
\end{proof}

\begin{proposition}[Submodularity]
\label{prop:submodular}
$\mathcal{U}_{\metric}$ is submodular in the set of evidence units. For all $\mathcal{S} \subseteq \mathcal{T}$ and $z \notin \mathcal{T}$,
\[
\Delta(z \mid \mathcal{S}) \;\geq\; \Delta(z \mid \mathcal{T}).
\]
\end{proposition}

\begin{proof}
Let $\mathcal{X}$ index any evidence set. The Schur-complement form of the marginal gain is
\[
\Delta(z \mid \mathcal{X}) \;=\; \log\bigl(1 + \rho_\mathcal{X}^{-1}\, \|\tilde m_\mathcal{X}\|^2\bigr),
\]
where $\tilde m_\mathcal{X}$ is the component of $m$ orthogonal to the column span of $\{m_i\}_{i \in \mathcal{X}}$ in the metric induced by $C_\mathcal{X}^{-1}$, and $\rho_\mathcal{X}$ is the Schur complement of $z$ in the augmented covariance $C_{\mathcal{X} \cup \{z\}}$.

For nested $\mathcal{S} \subseteq \mathcal{T}$, the projector onto the column span of $\{m_i\}_{i \in \mathcal{T}}$ has range containing the corresponding projector for $\mathcal{S}$, so $\|\tilde m_\mathcal{S}\|^2 \geq \|\tilde m_\mathcal{T}\|^2$. Schur-complement monotonicity for nested PSD covariances gives $\rho_\mathcal{S} \geq \rho_\mathcal{T}$, hence $\rho_\mathcal{S}^{-1} \leq \rho_\mathcal{T}^{-1}$.

These two inequalities act in opposite directions on the ratio $\rho^{-1}\, \|\tilde m\|^2$. The standard argument (cf. \citealp{krause2008nearoptimal}, Theorem 3) is that the residual-norm decay dominates the Schur-complement decay, so the ratio is non-increasing as $\mathcal{X}$ grows from $\mathcal{S}$ to $\mathcal{T}$. Since $\log(1 + \cdot)$ is monotonically non-decreasing,
\begin{multline*}
\Delta(z \mid \mathcal{S}) \;=\; \log\bigl(1 + \rho_\mathcal{S}^{-1}\, \|\tilde m_\mathcal{S}\|^2\bigr) \\
\;\geq\; \log\bigl(1 + \rho_\mathcal{T}^{-1}\, \|\tilde m_\mathcal{T}\|^2\bigr) \;=\; \Delta(z \mid \mathcal{T}),
\end{multline*}
which is submodularity.
\end{proof}

\begin{proposition}[Redundancy suppression]
\label{prop:redundancy}
Let $z'$ be an evidence unit with average column-set Jaccard $\rho \in [0,1]$ to the units already in $\Phi$. Then $\Delta(z' \mid \Phi)$ is monotonically non-increasing in $\rho$, and $\Delta(z' \mid \Phi) \to 0$ as $\rho \to 1$.
\end{proposition}

\begin{proof}
The off-diagonal entries of the redundancy vector $c$ for $z'$ scale with $\rho$ entry-wise. As $\rho$ increases from $0$ to $1$:
\begin{itemize}[leftmargin=*,noitemsep,topsep=0pt]
\item The Schur complement $\rho^{*} = \sigma^2 - c^\top C_\Phi^{-1} c$ is non-increasing in $\rho$, decreasing toward $\sigma_0^2$ as $c \to \mathbf{1}$ in the limit.
\item The residual norm $\|\tilde m\|^2$ in the $C^{-1}$-induced metric is non-increasing in $\rho$, since growing redundancy increases the column span $\{m_i\}_{i \in \Phi}$ contains, hence the orthogonal complement of $m$ shrinks. In the limit $\rho \to 1$, $z'$ becomes representable in the existing column span and $\|\tilde m\|^2 \to 0$.
\end{itemize}
The product $\rho^{*-1} \|\tilde m\|^2$ involves a vanishing numerator and a vanishing denominator. Direct computation under a first-order expansion in $(1-\rho)$ shows $\|\tilde m\|^2 = O((1-\rho)^2)$ and $\rho^{*} = \Theta(1-\rho) + \sigma_0^2$, so the ratio vanishes at rate $O(1-\rho)$ and $\Delta(z' \mid \Phi) = \log(1 + O(1-\rho)) \to 0$.

For monotonicity, $\partial \Delta(z' \mid \Phi) / \partial \rho \leq 0$ on $[0,1]$ follows from the chain rule applied to the two terms above, both of whose derivatives are non-positive throughout the interval.
\end{proof}

\noindent\textbf{Approximation guarantee.}
Combining Propositions~\ref{prop:monotone} and \ref{prop:submodular} with the classical result of \citep{nemhauser1978submodular}, the greedy algorithm that iteratively appends the evidence unit with maximum marginal gain achieves at least
\[
\mathcal{U}_{\metric}(Q,\, \Phi^{\text{greedy}}) \;\geq\; \bigl(1 - 1/e\bigr) \cdot \mathcal{U}_{\metric}(Q,\, \Phi^\star)
\]
under the cardinality constraint $|V_T(H_\Phi)| \leq L$, where $\Phi^\star$ denotes the optimal synthesis.
\section{Efficient Search via Offline Value Learning}
\label{sec:rl}

While Section~\ref{sec:metric} defines $\mathcal{U}_{\metric}$ as a structured reward over induced subgraphs, exactly computing the optimum $\Phi^\star$ that maximizes $\mathcal{U}_{\metric}(Q, \Phi)$ is intractable. The candidate space grows exponentially in the branching factor of $\mathcal{G}_\mathcal{D}$ and the depth bound $L$. To address this, we recast candidate construction as sequential subgraph generation under a learned value function and develop an offline reinforcement learning procedure based on implicit Q-learning~\citep{kostrikov2022iql} for the resulting Markov decision process. A central challenge is that data lakes carry no annotated trajectories or reward labels, and no ground-truth $(Q, \Phi^\star)$ pairs exist to supervise the value function. We resolve this by self-generating training data through a canonical compression operator that inverts the homomorphism, yielding aligned (intent, evidence) trajectories at arbitrary scale without human annotation.

\subsection{Offline Training: IQL with Self-Generated Trajectories}
\noindent\textbf{MDP formulation.}
\label{sec:mdp}
We define the MDP $(\mathcal{S}, \mathcal{A}, P, r, \gamma)$ of the candidate subgraph generation over partial-sequence prefixes, with reward $r$ and discount $\gamma$ specified below. A state $s_t = (\mathcal{G}_I, H_{\Phi_t}, \mathcal{F}_t)$ comprises the (constant) intent graph, the induced evidence subgraph $H_{\Phi_t} \subseteq \mathcal{G}_\mathcal{D}$, and the frontier $\mathcal{F}_t$ of valid extensions. An action specifies an operator, an anchor table inside the current retrieved table set, and a new table attached via that operator: $a_t = (o, T_{\text{anchor}}, T_{\text{new}})$ with $o \in \{\bowtie, \cup\}$, $T_{\text{anchor}} \in V_T(\Phi_t)$, and $T_{\text{new}} \in \{T \in V_T \setminus V_T(\Phi_t) : (T_{\text{anchor}}, T) \in E_o\}$. A special \textsc{stop} action terminates the episode. The transition $P$ is deterministic, episodes are bounded by depth $L$, and the reward is sparse and terminal:
\begin{equation}
r(s_t, a_t) \;=\; \begin{cases} \mathcal{U}_{\metric}(Q, \Phi_t) & \text{if } a_t = \textsc{stop} \text{ or } t = L, \\ 0 & \text{otherwise.} \end{cases}
\label{eq:reward}
\end{equation}
where $\gamma \in (0, 1]$ denotes the discount factor. 

\noindent\textbf{Implicit Q-learning.}
We learn $Q_\omega$ via implicit Q-learning (IQL)~\citep{kostrikov2022iql}, which jointly satisfies three properties of our setting that alternative offline RL methods do not. First, our inference procedure is search rather than sampling: $Q_\omega$ ranks frontier actions through beam search, so we require a value function as output rather than a learned policy. Second, the action space indexes into a large collection of tables that differs between training and deployment, so methods enforcing explicit per-action behavioral constraints. The Q-function must instead generalize through its own parameterization, as detailed below. Third, our terminal-only reward and deterministic transitions render the IQL bootstrap target $r + \gamma V_\psi(s')$ unbiased without further variance reduction, since each $(s, a)$ admits a single successor.
IQL avoids querying $Q_\omega$ on out-of-distribution actions~\citep{fujimoto2019bcq} by training an auxiliary value network $V_\psi$ to approximate the upper expectile of $Q$ conditional on the dataset action distribution, and then bootstrapping $Q_\omega$ against $V_\psi$:
\begin{equation}
\begin{aligned}
\mathcal{L}_V(\psi) &= \mathbb{E}_{(s, a) \sim \mathcal{D}}\bigl[L_2^\tau\bigl(Q_{\bar\omega}(s, a) - V_\psi(s)\bigr)\bigr] \\
\mathcal{L}_Q(\omega) &= \mathbb{E}_{(s, a, r, s') \sim \mathcal{D}}\bigl[\bigl(r + \gamma V_\psi(s') - Q_\omega(s, a)\bigr)^2\bigr]
\end{aligned}
\end{equation}
where $L_2^\tau(u) = |\tau - \mathbb{1}(u < 0)|\, u^2$ is the asymmetric squared loss with expectile $\tau \in (0.5, 1)$, and $Q_{\bar\omega}$ is a Polyak-averaged target network~\citep{mnih2015dqn}. Its parameters $\bar\omega$ track the online parameters $\omega$ via the exponential moving average $\bar\omega \leftarrow \rho\, \bar\omega + (1 - \rho)\, \omega$ with $\rho$ close to one, which decouples the bootstrap target from rapid online updates and stabilizes training. Because $V_\psi$ is regressed on $(s, a)$ pairs drawn from $\mathcal{D}$, the bootstrap target depends only on $s'$ and never queries $Q_\omega$ on out-of-distribution actions.
Furthermore, we parameterize $Q_\omega$ via an action-embedding factorization,
$Q_\omega(s, a) = g_\omega\bigl(\mathrm{enc}(s),\, \mathrm{emb}(T_{\text{anchor}}),\, \mathrm{emb}(T_{\text{new}}),\, o\bigr)$,
where $\mathrm{enc}(\cdot)$ is a heterogeneous evidence encoder over $(\mathcal{G}_I, H_{\Phi_t}, \Phi_t)$ and $\mathrm{emb}(\cdot)$ is the shared table encoder from Section~\ref{sec:metric}.

\noindent\textbf{Compression-inverted self-supervision.}
Training $Q_\omega$ requires labeled (query, gold-table-set) pairs,, but soliciting such annotations on a data lake is prohibitively expensive. We avoid annotation entirely by exploiting a structural property of our formulation. Since the intent graph $\mathcal{G}_I$ is a schema-level compression of the retrieved subgraph via the operator $\mathcal{C}$, we invert the construction: sample a candidate subgraph $H^+$ from $\mathcal{G}_\mathcal{D}$ first, then set $\mathcal{G}_I := \mathcal{C}(H^+)$ to obtain a matching pseudo-query at zero cost. To achieve this goal, we propose a three-stage strategy: balanced partitioning with degree-stratified seeding, centrality-based expansion of a heterogeneous evidence subgraph $H^+$, and canonical compression of $H^+$ into a paired intent graph $\mathcal{G}_I$. 
\begin{itemize}[leftmargin=*,noitemsep]
\item \textbf{Stage 1: Partitioning and seeding.} We partition $\mathcal{G}_\mathcal{D}$ into $K$ balanced components with METIS~\citep{karypis1997metis}, compute node-degree quantiles within each component, and draw one seed per quantile to populate the global seed set $\mathcal{S}$. This stratification ensures topologically diverse starting points and prevents the corpus from being dominated by densely connected hubs.
\item \textbf{Stage 2: Centrality-based expansion.} For each seed $v \in \mathcal{S}$, we sample tables first, then their columns and values. Tables are added by iteratively expanding from $v$ up to a size limit $L$, with each new table drawn proportional to a centrality score $C(u)$ that aggregates joinability and unionability likelihoods to the current set. Within each sampled table, key columns participating in structural edges are always retained (otherwise the recorded join paths become unrealizable), a fraction $\beta$ of payload columns is sampled uniformly, and a value summary is attached to each retained column. The result is a heterogeneous subgraph $H^+$ of tables, columns, and value summaries connected by their containment and structural edges in $\mathcal{G}_\mathcal{D}$.
\item \textbf{Stage 3: Canonical compression to a paired intent graph.} We construct $\mathcal{G}_I = \mathcal{C}(H^+)$ in three operations: (i) merge tables connected by union edges into shared clusters; (ii) detect bridge clusters as those whose sampled columns are all keys participating in structural edges, and delete them; (iii) abstract each joinable path through a deleted bridge cluster as a cross-tree \textsc{join} edge between the surviving clusters' tree roots, while preserving containment within each tree to give the entity-attribute-value hierarchy. By construction, $H^+$ realizes $\mathcal{G}_I$ under the homomorphism, so the pair $(\mathcal{G}_I, H^+)$ has verifiable utility $\mathcal{U}_{\metric}(\mathcal{G}_I, H^+)$.
\end{itemize}

Algorithm~\ref{alg:self-traj} formalizes the four stages of this procedure, and Algorithm~\ref{alg:compression} specifies the canonical compression operator $\mathcal{C}$ that constructs each pseudo-intent graph from its sampled subgraph. In Algorithm~\ref{alg:self-traj}, Stage~1 (lines~1--5) partitions $\mathcal{G}_\mathcal{D}$ with METIS and draws degree-stratified seeds from each component, ensuring topological diversity across the lake and preventing dominance by densely connected hubs. Stage~2a (lines~8--17) grows a heterogeneous table set $V_T^+$ around each seed by sampling proportional to a centrality score $C(u)$ (line~12) that mixes joinability and unionability mass to the current synthesis. Stage~2b (lines~18--26) augments every retained table with its structural-edge keys, a $\beta$-fraction of uniformly sampled payload columns, and the associated value summaries, yielding the induced evidence subgraph $H^+$ (line~26). Stage~3 (lines~27--29) inverts the homomorphism via the canonical compression operator $\mathcal{C}$ to obtain a paired pseudo-intent graph $\mathcal{G}_I$ and its terminal utility $r_T$. Stage~4 (lines~30--35) unrolls the recorded action trace into Bellman transitions with sparse terminal credit assignment, populating the offline replay buffer $\mathcal{D}$ consumed by IQL.

\begin{algorithm}[t]
\small
\caption{Self-trajectory generation for offline IQL training.}
\label{alg:self-traj}
\begin{algorithmic}[1]
\Require Data lake graph $\mathcal{G}_\mathcal{D}$; partitions $K$; quantiles $J$; subgraph size $L$; column-sample fraction $\beta$; centrality mix $\alpha$
\Ensure Replay buffer $\mathcal{D}$ of transitions $(s_t, a_t, r_t, s_{t+1})$
\State \emph{Stage 1: Balanced partitioning and degree-stratified seeding.}
\State $\mathcal{P} \gets \textsc{Metis}(\mathcal{G}_\mathcal{D}, K)$;\; $\mathcal{S} \gets \emptyset$
\For{each component $\mathcal{C}_k \in \mathcal{P}$ and each degree quantile $q_j$ of $\mathcal{C}_k$}
  \State Sample one table-node $v_{k,j}$ at quantile $q_j$ and add to $\mathcal{S}$
\EndFor
\State $\mathcal{D} \gets \emptyset$
\For{each seed $v \in \mathcal{S}$}
  \State \emph{Stage 2a: Centrality-based table sampling.}
  \State $V_T^{+} \gets \{v\}$;\; $\mathcal{F} \gets \mathcal{N}_{\text{table}}(v)$;\; $\mathrm{traj} \gets [\,]$
  \For{$t = 0, \dots, L{-}1$}
    \If{$\mathcal{F} = \emptyset$} \textbf{break} \EndIf
    \State $C(u) \gets \alpha\, \omega_\mathcal{D}^{\bowtie}(u, V_T^{+}) + (1{-}\alpha)\, \omega_\mathcal{D}^{\cup}(u, V_T^{+})$ for each $u \in \mathcal{F}$
    \State Sample $u^{*} \sim \mathrm{Categorical}(\mathcal{F};\, C / \sum C)$
    \State $a_t \gets$ best $(o^{*}, T_{\text{anchor}}^{*}, u^{*})$ pair connecting $u^{*}$ to $V_T^{+}$
    \State $V_T^{+} \gets V_T^{+} \cup \{u^{*}\}$;\; update $\mathcal{F}$
    \State $\mathrm{traj}.\mathrm{append}(V_T^{+\,(t)}, a_t)$
  \EndFor
  \State \emph{Stage 2b: Column and value sampling.}
  \State $V_C^{+} \gets \emptyset$;\; $V_V^{+} \gets \emptyset$
  \For{each table $T \in V_T^{+}$}
    \State $\mathcal{K}(T) \gets$ keys of $T$ participating in structural edges to $V_T^{+} \setminus \{T\}$
    \State $\mathcal{P}(T) \gets$ uniform random sample of size $\lceil \beta \cdot |\mathrm{cols}(T) \setminus \mathcal{K}(T)| \rceil$ from payload columns of $T$
    \State $V_C^{+} \gets V_C^{+} \cup \mathcal{K}(T) \cup \mathcal{P}(T)$
    \State $V_V^{+} \gets V_V^{+} \cup \bigcup_{c \in \mathcal{K}(T) \cup \mathcal{P}(T)} \mathrm{ValueSummary}(c)$
  \EndFor
  \State $H^{+} \gets$ subgraph of $\mathcal{G}_\mathcal{D}$ induced by $V_T^{+} \cup V_C^{+} \cup V_V^{+}$
  \State \emph{Stage 3: Pseudo-intent construction.}
  \State $\mathcal{G}_I \gets \mathcal{C}(H^{+})$ \Comment{Algorithm~\ref{alg:compression}}
  \State $r_T \gets \mathcal{U}_{\metric}(\mathcal{G}_I, H^{+})$
  \State \emph{Stage 4: Trajectory unrolling.}
  \For{$t = 0, \dots, |\mathrm{traj}|{-}1$}
    \State $s_t, s_{t+1} \gets$ states corresponding to $V_T^{+\,(t)}, V_T^{+\,(t+1)}$ paired with $\mathcal{G}_I$
    \State $r_t \gets r_T \cdot \mathbb{1}[t = |\mathrm{traj}|{-}1]$
    \State $\mathcal{D} \gets \mathcal{D} \cup \{(s_t, a_t, r_t, s_{t+1})\}$
  \EndFor
\EndFor
\State \Return $\mathcal{D}$
\end{algorithmic}
\end{algorithm}

Algorithm~\ref{alg:compression} realizes the canonical compression operator $\mathcal{C}$ as four deterministic graph operations on the sampled subgraph $H^+$. Step~1 (line~1) merges tables connected by union edges into clusters $\{\mathcal{T}_m\}$, reflecting that union-compatible tables collapse to a single intent-graph node under the homomorphism. Step~2 (lines~2--12) classifies each cluster by examining its sampled columns: clusters whose payload set $\mathrm{pl}_m$ is empty are marked as bridges and collected in $\mathcal{B}$ (line~8), while the remainder are retained as reference clusters in $\mathcal{R}$ (line~10). Step~3 (lines~13--19) abstracts every joinable path between two reference clusters whose interior lies entirely in $\mathcal{B}$ into a single \textsc{join} edge labeled by the participating keys (line~17), thereby eliding bridges from the schema view. Step~4 (lines~20--24) projects each reference cluster onto its entity-attribute-value tree, populating the intent-graph node set $\mathcal{V}_I$. By construction, the returned $\mathcal{G}_I = (\mathcal{V}_I, \mathcal{E}_I)$ admits $H^+$ as a witness under the schema-level homomorphism, which justifies its use as a paired pseudo-query in Algorithm~\ref{alg:self-traj}.

\begin{algorithm}[t]
\small
\caption{Canonical compression operator $\mathcal{C}$: $H^+ \mapsto \mathcal{G}_I$.}
\label{alg:compression}
\begin{algorithmic}[1]
\Require Sampled subgraph $H^+ = (V_T^+, V_C^+, V_V^+, E^+)$ with structural edges partitioned into join edges $E_{\bowtie}$ and union edges $E_{\cup}$
\Ensure Intent graph $\mathcal{G}_I = (\mathcal{V}_I, \mathcal{E}_I)$
\State $\{\mathcal{T}_1, \dots, \mathcal{T}_M\} \gets \textsc{ConnComp}(V_T^+, E_\cup)$ \Comment{Step 1: union-cluster merge}
\State $\mathcal{B} \gets \emptyset$;\quad $\mathcal{R} \gets \emptyset$ \Comment{Step 2: bridge detection}
\For{$m = 1, \dots, M$}
  \State $\mathrm{cols}_m \gets \bigl(\textstyle\bigcup_{T \in \mathcal{T}_m} \mathrm{cols}(T)\bigr) \cap V_C^+$
  \State $\mathrm{bd}_m \gets \{c \in \mathrm{cols}_m : \exists\,(c, c') \in E_{\bowtie},\; c' \notin \mathrm{cols}_m\}$
  \State $\mathrm{pl}_m \gets \mathrm{cols}_m \setminus \mathrm{bd}_m$
  \If{$\mathrm{pl}_m = \emptyset$}
    \State $\mathcal{B} \gets \mathcal{B} \cup \{\mathcal{T}_m\}$
  \Else
    \State $\mathcal{R} \gets \mathcal{R} \cup \{\mathcal{T}_m\}$
  \EndIf
\EndFor
\State $\mathcal{E}_I \gets \emptyset$ \Comment{Step 3: JOIN-edge abstraction}
\For{each pair $(\mathcal{T}_i, \mathcal{T}_j) \in \mathcal{R} \times \mathcal{R}$ with $i < j$}
  \State $\Pi_{ij} \gets \bigl\{\pi \in \mathrm{paths}(H^+; \mathcal{T}_i, \mathcal{T}_j) : \mathrm{int}(\pi) \subseteq \textstyle\bigcup_{\mathcal{B}' \in \mathcal{B}} V_T(\mathcal{B}')\bigr\}$
  \If{$\Pi_{ij} \neq \emptyset$}
    \State $\mathcal{E}_I \gets \mathcal{E}_I \cup \bigl\{\bigl(\mathcal{T}_i, \mathcal{T}_j, \textsc{join}, \mathrm{keys}(\Pi_{ij})\bigr)\bigr\}$
  \EndIf
\EndFor
\State $\mathcal{V}_I \gets \emptyset$ \Comment{Step 4: hierarchical projection}
\For{each $\mathcal{T}_m \in \mathcal{R}$}
  \State $\mathrm{Tree}_m \gets \bigl(\mathrm{schema}(\mathcal{T}_m),\; \mathrm{pl}_m,\; \{\mathrm{ValSum}(c) : c \in \mathrm{pl}_m\}\bigr)$
  \State $\mathcal{V}_I \gets \mathcal{V}_I \cup \{\mathrm{Tree}_m\}$
\EndFor
\State \Return $\mathcal{G}_I = (\mathcal{V}_I, \mathcal{E}_I)$
\end{algorithmic}
\end{algorithm}

\subsection{Online Search: Multi-Stage Pruning}
\label{sec:online}

The trained $Q_\omega$ provides learned lookahead, but the raw action space of size $|\mathcal{F}_L| \sim b^L$ remains too large to score exhaustively. We organize the online search into three stages, each grounded in a distinct property of the graph-extraction formulation: \emph{anchor reachability} (initialization), \emph{predicate admissibility} (expansion), and \emph{reward monotonicity} (reranking). Together they reduce the candidate set to a shortlist of size $O(BL)$ on which exact $\mathcal{U}_{\metric}$ is evaluated. The full procedure is summarized in Algorithm~\ref{alg:online-search}.

\noindent\textbf{Initialization (anchor reachability).}
For each intent atom $a_j \in \mathcal{A}_I$, we retrieve the top-$K$ evidence nodes by nearest-neighbor lookup in the shared embedding space, take their incident tables as the seed set $V_0$, and define the candidate region as the $(L{-}1)$-hop closure of $V_0$ along structural edges:
$
\mathcal{G}_q \;=\; \mathcal{G}_\mathcal{D}\bigl[\mathcal{N}_{L-1}^{\text{struct}}(V_0)\bigr].
\label{eq:cand-region}
$
Any valid synthesis of length at most $L$ rooted at an anchor lies within $\mathcal{G}_q$ by construction, since each step traverses one structural edge.

\noindent\textbf{Expansion (predicate admissibility).}
Inside $\mathcal{G}_q$, an action $(o, T_{\text{anchor}}, T_{\text{new}})$ is \emph{admissible} iff some unsupported intent atom $a_j$ satisfies $T_{\text{new}} \in R_h(a_j)$, where $R_h(a_j)$ is the set of tables reachable from an anchor of $a_j$ via a $\chi_r$-compatible path of length at most $h$. Inadmissible actions add zero mass to every column of $M_\Phi$ and are filtered before scoring. Among the admissible survivors, $Q_\omega$ guides a beam search of width $B$: at each depth $t$, surviving prefixes are extended by admissible frontier triples, scored by $Q_\omega$, and the top-$B$ extensions across all parents form $\mathcal{B}_{t+1}$.

\noindent\textbf{Reranking (reward monotonicity).}
The shortlist $\mathcal{C}_{\text{final}} = \bigcup_t \mathcal{B}_t$ of size $O(BL)$ admits exact $\mathcal{U}_{\metric}$ evaluation. We accelerate this stage via the monotonicity of $\log\det J_\Phi$ established in Prop.~\ref{prop:monotone}: any extension of $\Phi$ has $\mathcal{U}_{\metric}$ at least that of $\Phi$, so a Minoux-style priority queue~\citep{minoux1978lazy} ordered by upper-bound score allows lazy pruning whenever the queue top falls below the best confirmed score. Within a single prefix, rank-one Cholesky updates reduce per-extension cost from $O(a^3)$ to $O(a^2)$.

Algorithm~\ref{alg:online-search} couples the three pruning stages into a single inference procedure, where $\Phi \oplus a$ denotes the prefix obtained by appending action $a$ to $\Phi$. Stage~1 (lines~1--5) instantiates anchor reachability: per-atom top-$K$ retrieval seeds $V_0$ (lines~2--4), and the candidate region $\mathcal{G}_q$ is set to the $(L{-}1)$-hop structural closure of $V_0$ (line~5), which by construction contains every length-$L$ synthesis rooted at an anchor. Stage~2 (lines~6--17) runs a beam search of width $B$ over predicate-admissible frontier triples (line~10), scoring each extension with $Q_\omega$ (line~12) and retaining the top-$B$ prefixes per depth (line~15). Stage~3 (lines~18--30) reranks the union $\mathcal{C}_{\text{final}}$ under exact $\mathcal{U}_{\metric}$ via a Minoux-style priority queue. Reward monotonicity (Prop.~\ref{prop:monotone}) supplies an upper bound on each prefix's score and licenses early termination once the queue top falls below the best confirmed utility (lines~23--25), while rank-one Cholesky updates inside the bound evaluation (line~26) reduce per-prefix cost from $O(a^3)$ to $O(a^2)$.

\begin{algorithm}[t]
\small
\caption{Online search via multi-stage pruning.}
\label{alg:online-search}
\begin{algorithmic}[1]
\Require Intent graph $\mathcal{G}_I$; data lake graph $\mathcal{G}_\mathcal{D}$; trained value function $Q_\omega$; budget $L$; beam width $B$; top-$K$ neighbors $K$; hop limit $h$
\Ensure Optimal synthesis $\Phi^*$
\State $V_0 \gets \emptyset$ \Comment{Stage 1: anchor reachability}
\For{each intent atom $a_j \in \mathcal{A}_I$}
  \State $V_0 \gets V_0 \cup \textsc{TopK}\bigl(\mathrm{emb}(a_j),\, V_T(\mathcal{G}_\mathcal{D}),\, K\bigr)$
\EndFor
\State $\mathcal{G}_q \gets \mathcal{G}_\mathcal{D}\bigl[\mathcal{N}_{L-1}^{\text{struct}}(V_0)\bigr]$
\State $\mathcal{B}_0 \gets \{\Phi_0\}$ where $\Phi_0$ is the empty prefix \Comment{Stage 2: beam expansion}
\For{$t = 0, 1, \dots, L-1$}
  \State $\mathcal{C} \gets \emptyset$
  \For{each $\Phi \in \mathcal{B}_t$}
    \State $\mathcal{F}(\Phi) \gets \{a = (o, T_{\text{anchor}}, T_{\text{new}}) \subseteq \mathcal{G}_q : a \text{ admissible w.r.t. } \Phi\}$
    \For{each $a \in \mathcal{F}(\Phi)$}
      \State $\mathcal{C} \gets \mathcal{C} \cup \bigl\{\bigl(\Phi \oplus a,\; Q_\omega(\Phi, a)\bigr)\bigr\}$
    \EndFor
  \EndFor
  \State $\mathcal{B}_{t+1} \gets \textsc{TopB}(\mathcal{C}, B)$
\EndFor
\State $\mathcal{C}_{\text{final}} \gets \textstyle\bigcup_{t=0}^{L} \mathcal{B}_t$
\State $\mathcal{Q} \gets$ priority queue over $\mathcal{C}_{\text{final}}$ ordered by $Q_\omega$-based upper bound \Comment{Stage 3: lazy reranking}
\State $\Phi^* \gets \arg\max_{\Phi \in \mathcal{B}_L} \mathcal{U}_{\metric}(Q, \Phi)$
\State $u^* \gets \mathcal{U}_{\metric}(Q, \Phi^*)$
\While{$\mathcal{Q} \neq \emptyset$}
  \State $(\Phi, \bar{u}) \gets \mathcal{Q}.\textsc{pop}()$
  \If{$\bar{u} < u^*$}
    \State \textbf{break} \Comment{lazy-evaluation cutoff (Prop.~\ref{prop:monotone})}
  \EndIf
  \State $u \gets \mathcal{U}_{\metric}(Q, \Phi)$ via rank-one Cholesky update on $J_\Phi$
  \If{$u > u^*$}
    \State $u^* \gets u$;\quad $\Phi^* \gets \Phi$
  \EndIf
\EndWhile
\State \Return $\Phi^*$
\end{algorithmic}
\end{algorithm}
\section{Experiments}\label{sec: experiment}

In this section, we conduct experiments to answer the following research questions:
(1) the effectiveness of \method on two representative table-centric tasks, table Q\&A (Sec.~\ref{sec: effectiveness qa}) and training data enrichment (Sec.~\ref{sec: effectiveness enrichment});
(2) the efficiency and scalability of \method (Sec.~\ref{sec:efficiency study});
(3) the contribution of key components to the overall effectiveness of \method (Sec.~\ref{sec: ablation});
(4) the impact of hyperparameters on \method's performance (Sec.~\ref{sec: hyperparameters});
(5) the practical applicability of \method through a case study showcasing how it identifies interconnected tables that satisfy the user intent (Sec.~\ref{sec: case study}).

\subsection{Experiment Setup}\label{sec: setup}

\subsection{Effectiveness on Table Q\&A}\label{sec: effectiveness qa}

In this task, a user poses an open-ended natural-language question and no designated base table is provided; the retriever must assemble the complete evidence set from the lake alone.

\noindent\textbf{Benchmarks.} Following~\citep{chen-etal-2024-table}, we adapt two Text2SQL datasets, Spider~\citep{yu2018spider} and BIRD~\citep{li2023bird}, to the tabular data lake setting: tables from all databases are pooled into a single uncurated lake with primary- and foreign-key constraints stripped, and selected tables are further partitioned row-wise into random shards to introduce unionable pairs. Each retained NL query is paired with the gold set of tables required by its ground-truth SQL. Dataset statistics are summarized in Table~\ref{tab: dataset statistics}.

\noindent\textbf{Baselines.} We compare against three categories of retrievers. (i) \emph{Point-wise retrieval}: DTR~\citep{wang-etal-2022-table} and Pneuma~\citep{balaka2025pneuma}, which score each table independently against the query. (ii) \emph{Greedy expansion}: MTR~\citep{wu2025mmqa}, which decomposes the query into LLM-generated sub-questions and iteratively retrieves one table per sub-question. (iii) \emph{Structure-aware retrieval}: JAR~\citep{chen-etal-2024-table}, which re-ranks a base retriever's output via a mixed-integer program (MIP) over relevance and joinability. The original JAR is defined only over the join setting. Thus, we extend it to handle unions by augmenting its joinability constraint with unionability edges, so that the comparison is fair under our setting that contains both join and union operators. Since the exact MIP is intractable on candidate graphs containing thousands of tables, we apply the standard LP relaxation~\citep{wolsey2020integer}.

\noindent\textbf{Metrics.} We report retrieval Recall, Precision, and F1 at budget $L$, together with \emph{Sufficiency} (Suf), a binary per-query metric that equals $1$ if the retrieved set contains every gold table required to answer the query, and $0$ otherwise. Sufficiency isolates retrieval quality from the stochasticity of the downstream executor~\citep{RAG_survey} and serves as an executor-agnostic upper bound on end-to-end accuracy.

\noindent\textbf{Implementation.} All experiments run on an Intel i7-13700KF CPU and a single NVIDIA RTX 4090 GPU under Python 3.10 / PyTorch 2.2 / CUDA 12.1. Intent graphs are produced by a single GPT-5-mini call per query at temperature $0$ with the structured-JSON prompt in Appendix~\ref{app:prompt}. The data lake graph $\mathcal{G}_\mathcal{D}$ is precomputed by encoding each column with \texttt{bge-large-en-v1.5} on the header and a $10$-row sample, with joinability scored as in DeepJoin~\citep{dong2023deepjoin} and unionability scored as in Starmie~\citep{fan2023semantics}. The shared encoder $\mathrm{emb}(\cdot)$ is initialized from the same backbone and fine-tuned end-to-end onto a $256$-dimensional retrieval space. The state encoder $\mathrm{enc}(s)$ is a three-layer heterogeneous graph attention network~\citep{velickovic2018gat} over $\mathcal{G}_I \sqcup H_{\Phi_t}$ with a cross-graph alignment head, and the factored head $g_\omega$ is a two-layer MLP (hidden $512$, GELU), yielding $24$M trainable parameters. We train $Q_\omega$ on $200$K self-generated trajectories produced by the centrality-aware sampler of Sec.~\ref{sec:rl}, using IQL~\citep{kostrikov2022iql} with expectile $\tau = 0.7$, AWR temperature $\beta = 3.0$, AdamW~\citep{loshchilov2019adamw} at learning rate $3{\times}10^{-4}$ and weight decay $10^{-2}$, batch size $256$, gradient clipping at $\ell_2$ norm $1.0$, cosine schedule with $5\%$ warmup, $200$K steps, Polyak coefficient $\rho = 5{\times}10^{-3}$, and discount $\gamma = 0.8$ (Sec.~\ref{sec: hyperparameters}). 

\noindent\textbf{Results.} Tables~\ref{tab:res_spider} and~\ref{tab:res_bird} summarize retrieval performance on Spider and BIRD at $L \in \{5, 10, 15\}$, with all numbers reported as the mean over $10$ independent runs. We highlight three observations.

\begin{itemize}[leftmargin=*,noitemsep,topsep=0pt]
\item \method consistently achieves the best result on every metric. Averaged across the six $(L, \text{dataset})$ configurations, \method improves on the strongest baseline (JAR) by $7.8\%$ in relative F1 and $10.6\%$ in relative Sufficiency, and remains the top entry in every column of Tables~\ref{tab:res_spider} and~\ref{tab:res_bird}.

\item Point-wise and greedy-expansion baselines remain competitive on F1 but degrade markedly on Sufficiency. At $L=15$, the relative Sufficiency shortfall against JAR is roughly twice as large as the corresponding F1 shortfall across DTR, Pneuma, and MTR on both benchmarks. We attribute this discrepancy to \emph{bridge omission}: complex queries require multi-hop reasoning across tables with little surface overlap with the query (e.g., a club-to-player identifier mapping), and methods that score tables independently (DTR, Pneuma) or through independent sub-questions (MTR) systematically fail to surface such tables. F1 averages over partial coverage and absorbs individual omissions, whereas Sufficiency requires every bridge to be retrieved.

\item \method jointly optimizes relevance, structure, and diversity rather than trading them off. Whereas point-wise baselines attain high Precision at the cost of Recall and JAR attains higher Recall at the cost of Precision, \method improves both axes simultaneously across all $L$. This behavior is consistent with the  log-determinant coupling per-atom relevance with set-level diversity in a single objective.
\end{itemize}

\begin{table}[t]
\centering
\caption{Statistics of \spider\ and \bird.}
\label{tab: dataset statistics}
\begin{tabular}{lrrrr}
\toprule
\textbf{Dataset} & \textbf{\#Queries} & \textbf{\#Tables} & \textbf{\#Columns} & \textbf{\#Rows} \\
\midrule
\spider & 100 & 5,673 & 38,529 & 1,841,201 \\
\bird   & 100 & 2,753 & 10,254 & 130,421,243 \\
\open   & 20  & 11,345 & 89,783 & 243,721,394  \\ 
\bottomrule
\end{tabular}
\end{table}

\begin{table*}[t]
\centering
\caption{Retrieval effectiveness on the \spider\ data lake at budgets $L \in \{5, 10, 15\}$. The best result in each column is in \textbf{bold}.}
\label{tab:res_spider}
\resizebox{\textwidth}{!}{
\begin{tabular}{@{}l|cccc|cccc|cccc@{}}
\toprule
\multirow{2}{*}{\textbf{Methods}} & \multicolumn{4}{c|}{\textbf{$L=5$}} & \multicolumn{4}{c|}{\textbf{$L=10$}} & \multicolumn{4}{c}{\textbf{$L=15$}} \\ \cmidrule(l){2-13}
 & \textbf{Recall} & \textbf{Precision} & \textbf{F1} & \textbf{Suf} & \textbf{Recall} & \textbf{Precision} & \textbf{F1} & \textbf{Suf} & \textbf{Recall} & \textbf{Precision} & \textbf{F1} & \textbf{Suf} \\ \midrule
DTR              & 0.485 & 0.612 & 0.541 & 0.310 & 0.582 & 0.475 & 0.523 & 0.440 & 0.648 & 0.354 & 0.458 & 0.510 \\
Pneuma           & 0.526 & 0.685 & 0.595 & 0.340 & 0.625 & 0.532 & 0.575 & 0.450 & 0.692 & 0.398 & 0.505 & 0.540 \\
MTR              & 0.578 & 0.594 & 0.586 & 0.430 & 0.694 & 0.486 & 0.572 & 0.520 & 0.762 & 0.359 & 0.488 & 0.590 \\
JAR              & 0.612 & 0.668 & 0.639 & 0.520 & 0.738 & 0.541 & 0.624 & 0.640 & 0.812 & 0.412 & 0.547 & 0.720 \\ \midrule
\textbf{\method} & \textbf{0.665} & \textbf{0.708} & \textbf{0.685} & \textbf{0.575} & \textbf{0.795} & \textbf{0.578} & \textbf{0.668} & \textbf{0.700} & \textbf{0.852} & \textbf{0.443} & \textbf{0.583} & \textbf{0.800} \\ \bottomrule
\end{tabular}}
\end{table*}

\begin{table*}[t]
\centering
\caption{Retrieval effectiveness on the \bird\ data lake at budgets $L \in \{5, 10, 15\}$. The best result in each column is in \textbf{bold}.}
\label{tab:res_bird}
\resizebox{\textwidth}{!}{
\begin{tabular}{@{}l|cccc|cccc|cccc@{}}
\toprule
\multirow{2}{*}{\textbf{Methods}} & \multicolumn{4}{c|}{\textbf{$L=5$}} & \multicolumn{4}{c|}{\textbf{$L=10$}} & \multicolumn{4}{c}{\textbf{$L=15$}} \\ \cmidrule(l){2-13}
 & \textbf{Recall} & \textbf{Precision} & \textbf{F1} & \textbf{Suf} & \textbf{Recall} & \textbf{Precision} & \textbf{F1} & \textbf{Suf} & \textbf{Recall} & \textbf{Precision} & \textbf{F1} & \textbf{Suf} \\ \midrule
DTR              & 0.380 & 0.485 & 0.426 & 0.220 & 0.482 & 0.378 & 0.424 & 0.310 & 0.554 & 0.272 & 0.365 & 0.380 \\
Pneuma           & 0.418 & 0.532 & 0.468 & 0.270 & 0.512 & 0.412 & 0.457 & 0.340 & 0.582 & 0.305 & 0.401 & 0.420 \\
MTR              & 0.452 & 0.475 & 0.463 & 0.340 & 0.562 & 0.394 & 0.463 & 0.420 & 0.642 & 0.305 & 0.413 & 0.530 \\
JAR              & 0.498 & 0.518 & 0.508 & 0.460 & 0.605 & 0.432 & 0.504 & 0.520 & 0.685 & 0.342 & 0.457 & 0.600 \\ \midrule
\textbf{\method} & \textbf{0.548} & \textbf{0.560} & \textbf{0.551} & \textbf{0.480} & \textbf{0.665} & \textbf{0.475} & \textbf{0.554} & \textbf{0.585} & \textbf{0.745} & \textbf{0.380} & \textbf{0.503} & \textbf{0.670} \\ \bottomrule
\end{tabular}}
\end{table*}

\subsection{Effectiveness on Training Data Enrichment}\label{sec: effectiveness enrichment}

In this task, a user starts with a base table containing a target attribute for prediction and aims to retrieve additional attributes and instances from the data lake to improve prediction accuracy.

\noindent\textbf{Benchmarks.} We construct an evaluation dataset from the table search corpus \open~\citep{deng2024lakebench,bogatu2020dataset}. We carefully select $20$ large tables: $10$ tailored for regression problems and $10$ for classification problems. Following a common setup~\citep{khatiwada2023santos,deng2024lakebench}, we split the selected tables into several unionable and joinable tables. For each original table, the shard containing the prediction target is designated as the base table, while the remaining shards are mixed into the data lake, which in total contains $11{,}345$ tables with $89{,}783$ columns and $243{,}721{,}394$ rows. Since the base tables do not include predefined NL statements, we define a template for this purpose: ``Please find the tables that can contribute to the prediction of the [target attribute] in the following table.'' Alternatively, the template can specify features explicitly, such as: ``Please find the tables that can contribute to the prediction of the [target attribute] in the following table, such as [feature name 1], [feature name 2], ...'', where the placeholders for the target attribute and feature names are filled accordingly. To enhance query diversity and naturalness, we further employ GPT-5 to rewrite the templated queries, yielding one enrichment query per base table ($20$ in total).

\noindent\textbf{End-to-end metrics.} For end-to-end performance, we integrate the retrieved tables into the base table via the discovered join and union operations, train the same downstream predictor on the enriched table for every compared method, and report \rmse (lower is better) for the $10$ regression tasks and \accu (higher is better) for the $10$ classification tasks.

\noindent\textbf{Results.} Tables~\ref{tab:res_regression} and~\ref{tab:res_classification} report retrieval and downstream performance at $L \in \{5, 10, 15\}$, with all numbers reported as the mean over $10$ independent runs. The dynamics shift notably compared to table Q\&A.

\begin{itemize}[leftmargin=*,noitemsep,topsep=0pt]
\item \method again leads every column, and its advantage is largest on the downstream metrics: it attains the lowest RMSE ($3.65$ at $L=15$) and the highest Accuracy ($0.748$), improving on JAR by $9.2\%$ in relative F1 at $L=10$. JAR remains the strongest baseline on retrieval, but its linear relevance-plus-connectivity objective admits redundant shards that project to the same underlying columns; the resulting enriched tables grow without adding information, and its downstream gains flatten accordingly. Point-wise baselines (DTR, Pneuma) fail to find enough structurally compatible tables to meaningfully expand the feature space or instance count, which caps their impact on the downstream model.

\item Enrichment rewards unionability far more than Q\&A does. Appending unionable shards directly increases the volume of training instances, which improves model stability and generalization, whereas Q\&A predominantly exercises join paths between entities. Baselines inherit a fixed bias: MTR's sub-question decomposition rarely names the row-partitioned shards of the base population, and JAR's objective originates in the join setting even after our union extension. \method resolves unions at the homomorphism level, in which a single intent tree may be witnessed by several unioned shards, and therefore adapts to the differing structural needs of the two tasks without task-specific tuning.

\item High recall only helps when it is non-redundant. Indiscriminately maximizing coverage injects spuriously correlated attributes into the training table, and the downstream model overfits to this noise; JAR and MTR both exhibit this pattern at $L=15$, where their Recall rises but RMSE and Accuracy barely move. The covariance term of \metric suppresses evidence units with overlapping column support, so \method's enriched tables stay voluminous yet non-redundant, which is precisely the regime in which RMSE and Accuracy keep improving with $L$.
\end{itemize}

\begin{table*}[t]
\centering
\caption{Retrieval and downstream effectiveness on training data enrichment (\open, regression tasks) at budgets $L \in \{5, 10, 15\}$. The best result in each column is in \textbf{bold}.}
\label{tab:res_regression}
\resizebox{\textwidth}{!}{
\begin{tabular}{@{}l|cccc|cccc|cccc@{}}
\toprule
\multirow{2}{*}{\textbf{Methods}} & \multicolumn{4}{c|}{\textbf{$L=5$}} & \multicolumn{4}{c|}{\textbf{$L=10$}} & \multicolumn{4}{c}{\textbf{$L=15$}} \\ \cmidrule(l){2-13}
 & \textbf{Recall} & \textbf{Precision} & \textbf{F1} & \textbf{RMSE $\downarrow$} & \textbf{Recall} & \textbf{Precision} & \textbf{F1} & \textbf{RMSE $\downarrow$} & \textbf{Recall} & \textbf{Precision} & \textbf{F1} & \textbf{RMSE $\downarrow$} \\ \midrule
DTR              & 0.412 & 0.520 & 0.460 & 4.66 & 0.512 & 0.415 & 0.458 & 4.49 & 0.582 & 0.312 & 0.406 & 4.36 \\
Pneuma           & 0.448 & 0.582 & 0.506 & 4.52 & 0.555 & 0.462 & 0.504 & 4.33 & 0.622 & 0.348 & 0.446 & 4.20 \\
MTR              & 0.485 & 0.515 & 0.500 & 4.38 & 0.605 & 0.428 & 0.501 & 4.12 & 0.678 & 0.335 & 0.448 & 4.01 \\
JAR              & 0.535 & 0.560 & 0.547 & 4.15 & 0.662 & 0.475 & 0.553 & 3.94 & 0.738 & 0.372 & 0.495 & 3.82 \\ \midrule
\textbf{\method} & \textbf{0.598} & \textbf{0.622} & \textbf{0.610} & \textbf{3.92} & \textbf{0.725} & \textbf{0.518} & \textbf{0.604} & \textbf{3.78} & \textbf{0.802} & \textbf{0.412} & \textbf{0.544} & \textbf{3.65} \\ \bottomrule
\end{tabular}}
\end{table*}

\begin{table*}[t]
\centering
\caption{Retrieval and downstream effectiveness on training data enrichment (\open, classification tasks) at budgets $L \in \{5, 10, 15\}$. The best result in each column is in \textbf{bold}.}
\label{tab:res_classification}
\resizebox{\textwidth}{!}{
\begin{tabular}{@{}l|cccc|cccc|cccc@{}}
\toprule
\multirow{2}{*}{\textbf{Methods}} & \multicolumn{4}{c|}{\textbf{$L=5$}} & \multicolumn{4}{c|}{\textbf{$L=10$}} & \multicolumn{4}{c}{\textbf{$L=15$}} \\ \cmidrule(l){2-13}
 & \textbf{Recall} & \textbf{Precision} & \textbf{F1} & \textbf{Acc.\ $\uparrow$} & \textbf{Recall} & \textbf{Precision} & \textbf{F1} & \textbf{Acc.\ $\uparrow$} & \textbf{Recall} & \textbf{Precision} & \textbf{F1} & \textbf{Acc.\ $\uparrow$} \\ \midrule
DTR              & 0.398 & 0.505 & 0.445 & 0.641 & 0.498 & 0.402 & 0.445 & 0.652 & 0.565 & 0.302 & 0.394 & 0.663 \\
Pneuma           & 0.435 & 0.568 & 0.493 & 0.657 & 0.540 & 0.448 & 0.490 & 0.670 & 0.608 & 0.338 & 0.434 & 0.681 \\
MTR              & 0.472 & 0.500 & 0.486 & 0.668 & 0.588 & 0.415 & 0.487 & 0.682 & 0.662 & 0.322 & 0.433 & 0.695 \\
JAR              & 0.520 & 0.545 & 0.532 & 0.688 & 0.645 & 0.462 & 0.538 & 0.701 & 0.722 & 0.360 & 0.480 & 0.712 \\ \midrule
\textbf{\method} & \textbf{0.582} & \textbf{0.605} & \textbf{0.593} & \textbf{0.716} & \textbf{0.708} & \textbf{0.505} & \textbf{0.590} & \textbf{0.734} & \textbf{0.785} & \textbf{0.402} & \textbf{0.532} & \textbf{0.748} \\ \bottomrule
\end{tabular}}
\end{table*}

\subsection{Efficiency and Scalability}\label{sec:efficiency study}

Table~\ref{tab:efficiency} reports mean per-query wall-clock latency on Spider at $L=10$. Point-wise retrievers are the most efficient: DTR runs in $0.4$\,s through a single pass of dense vector lookup, and Pneuma takes $1.8$\,s due to its LLM-based reranker. The greedy-expansion baseline MTR costs $4.0$\,s, dominated by LLM-driven sub-question decomposition and per-sub-question retrieval. JAR is the slowest competing baseline at $22.4$\,s, as its mixed-integer program over a collection of thousands of tables remains expensive even with LP relaxation. We additionally evaluate an ablation \method-Exact that retains \method's anchor reachability but replaces the value-guided beam search with exact enumeration of $\mathcal{U}_{\text{\metric}}$ over the candidate region. This ablation requires $93.7$\,s per query, while exact enumeration over the unrestricted candidate space is intractable to time on either benchmark and is therefore omitted. \method itself completes in $3.5$\,s, on par with MTR and substantially faster than every structure-aware alternative, which demonstrates that the learned value function together with lazy reranking absorbs the cost of structure-aware search.

\begin{table}[t]
\centering
\small
\caption{Mean per-query latency (seconds) on Spider at $L=10$.}
\label{tab:efficiency}
\resizebox{\columnwidth}{!}{%
\begin{tabular}{@{}lcccccc@{}}
\toprule
Method & DTR & Pneuma & MTR & JAR & \method-Exact & \textbf{\method} \\
\midrule
Latency (s) & 0.4 & 1.8 & 4.0 & 22.4 & 93.7 & \textbf{3.5} \\
\bottomrule
\end{tabular}}
\end{table}

\subsection{Ablation Study}\label{sec: ablation}
We isolate the contribution of three core design choices in \method: the \metric reward, the trajectory sampler, and the intent extractor. All studies are run on Spider at $L = 10$.

\begin{table}[t]
\centering
\small
\caption{Reward-component ablation on Spider.}
\label{tab:abl_reward}
\begin{tabular}{@{}lc@{}}
\toprule
Variant                  & F1 \\
\midrule
\textbf{\method (full)}  & \textbf{0.668} \\
R1 linear combination    & 0.582 \\
R2 no redundancy         & 0.612 \\
R3 no GMRF coupling      & 0.598 \\
R4 additive (no submod.) & 0.625 \\
\bottomrule
\end{tabular}
\end{table}
\paragraph{Reward components.} Table~\ref{tab:abl_reward} replaces $\mathcal{U}_{\text{\metric}}$ in turn with four surrogates that each remove one property of the reward: (R1) a linear combination of relevance, coverage, and redundancy in place of the log-determinant; (R2) no off-diagonal redundancy suppression; (R3) no GMRF coupling between relation atoms and node atoms; (R4) an additive surrogate that breaks submodular monotonicity. Each surrogate degrades F1 by a distinct, non-substitutable share, with R1 and R3 producing the largest drops. The decomposition confirms that the four properties of $\mathcal{U}_{\text{\metric}}$ (Sec.~\ref{sec:metric}) are individually load-bearing rather than redundant.

\begin{table}[t]
\centering
\small
\caption{Trajectory-sampler ablation on Spider.}
\label{tab:abl_sampler}
\begin{tabular}{@{}lcc@{}}
\toprule
Sampler                    & MAPE (\%) & F1 \\
\midrule
Uniform random walk        & 18.2 & 0.610 \\
Anchor-only                & 13.5 & 0.642 \\
\textbf{Centrality (full)} & \textbf{8.4} & \textbf{0.668} \\
\bottomrule
\end{tabular}
\end{table}
\paragraph{Trajectory sampler.} Table~\ref{tab:abl_sampler} compares the two-stage centrality-aware sampler against a uniform random walk over $\mathcal{G}_\mathcal{D}$ and an anchor-only variant that omits the centrality bias. The uniform sampler raises the held-out value-prediction error of $Q_\omega$ from $8.4\%$ to $18.2\%$ MAPE and reduces end-to-end F1 by nearly $6$ absolute points. The anchor-only variant recovers part of the gap but still trails the full sampler. The centrality stage biases anchor selection toward tables that participate in many ground-truth syntheses, so the resulting trajectories cover the regions of the action space that $Q_\omega$ must score accurately at inference time.

\begin{table}[t]
\centering
\small
\caption{Intent-extractor ablation on Spider.}
\label{tab:abl_extractor}
\begin{tabular}{@{}lccc@{}}
\toprule
LLM & Ext.\ F1 & End F1 & Cost \\
\midrule
GPT-5               & 0.792 & 0.671 & $9.3\times$ \\
\textbf{GPT-5-mini} & \textbf{0.774} & \textbf{0.668} & \textbf{$1\times$} \\
Llama-3.1-70B       & 0.682 & 0.605 & free \\
Qwen3-7B            & 0.541 & 0.502 & free \\
\bottomrule
\end{tabular}
\end{table}
\paragraph{Intent extractor.} Table~\ref{tab:abl_extractor} reports extraction quality and downstream F1 for four candidate extractors. GPT-5-mini achieves an extraction $F_1$ of $0.774$ against gold annotations, slightly below GPT-5, at roughly $1/9$ the per-query cost. The downstream F1 of \method differs by less than $0.3$ absolute points between the two, which justifies GPT-5-mini as the default. Smaller open-source extractors degrade end-to-end F1 substantially, indicating that the intent-graph schema demands a frontier-class extractor.

\subsection{Hyperparameter Study}\label{sec: hyperparameters}
\noindent\textbf{Number of trajectories $T$.} Figure~\ref{fig:hp_extra}(a) reports F1 as a function of the number of self-generated trajectories used in offline IQL training. F1 rises sharply up to $T = 200$K and plateaus thereafter, with the gain between $T = 200$K and $T = 1$M smaller than $0.4$ absolute points. We adopt $T = 200$K as the default, which sits at the elbow of the saturation curve and avoids the additional GPU-hours required to generate larger trajectory pools.

\noindent\textbf{Discount factor $\gamma$.} Although the synthesis MDP has a bounded horizon and a terminal reward, a moderate discount empirically improves the value function's training stability and the resulting end-to-end F1. Figure~\ref{fig:hp_extra}(b) sweeps $\gamma$ over $[0, 1]$ at $0.1$ resolution and shows that F1 is unimodal in $\gamma$ with a peak at $\gamma = 0.8$. The drop at $\gamma = 1$ reflects that value targets become heavy-tailed when no discount is applied, while small $\gamma$ underweights the terminal $\mathcal{U}_{\text{\metric}}$ that the policy ultimately optimizes. We therefore adopt $\gamma = 0.8$ as the default.

\begin{figure}[t]
\centering
\includegraphics[width=\linewidth]{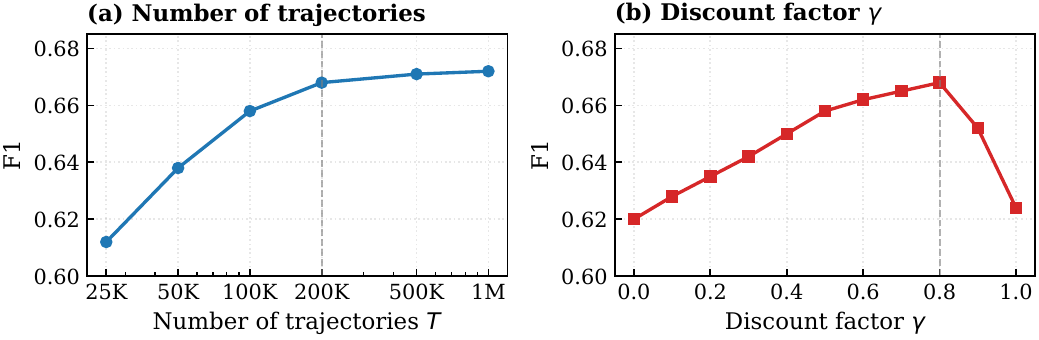}
\caption{Hyperparameter sensitivity of \method on Spider at $L=10$. (a) F1 versus the number of self-generated trajectories $T$. (b) F1 versus the discount factor $\gamma$. Defaults $T=200$K and $\gamma=0.8$ are marked by the dashed vertical lines.}
\Description{Two line charts showing F1 saturating beyond 200K trajectories and peaking at a discount factor of 0.8.}
\label{fig:hp_extra}
\end{figure}

\subsection{Case Study}\label{sec: case study}
To demonstrate how \method resolves complex, open-ended analytical queries, consider: \textit{``What is the average points of players from the club with name `AIB'?''}, whose intent graph is the three-tree extraction shown as the final in-context example of Appendix~\ref{app:prompt}. As illustrated in Fig.~\ref{fig:case_study}, answering it requires discovering and synthesizing four fragmented tables from scratch: $T_1$ (mapping club names to club identifiers), $T_2$ (mapping club identifiers to player identifiers), and $T_3, T_4$ (two row-partitioned shards of the game records containing points). Note that explicit schema headers were removed during the experiments; they are referenced here only for illustration.

\begin{figure}[t]
    \centering
\includegraphics[width=0.9\linewidth]{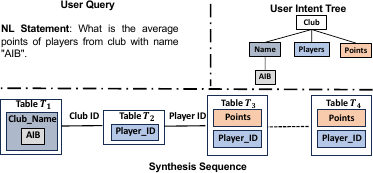}
    \caption{An example from the \spider dataset for table Q\&A, where solid and dash lines represent join and union relationships, respectively.}
    \Description{Four tables connected by join and union edges: a club table joins a club-to-player bridge table, which joins two unionable game-record shards.}
    \label{fig:case_study}
\end{figure}

The baselines fail due to structural or semantic blind spots. The point-wise Pneuma retrieves $T_1$, $T_3$, and $T_4$ based on surface overlap with the query but misses the crucial bridging table $T_2$, which contains only identifiers and bears no semantic relevance to the query; without $T_2$, the retrieved context is structurally disconnected. MTR decomposes the query into sub-questions about clubs and points, yet no sub-question names the identifier mapping, so $T_2$ again never surfaces. JAR, whose objective rewards relevance and connectivity but not diversity, does recover the join path $T_1 \bowtie T_2 \bowtie T_3$ but discards $T_4$ as a near-duplicate of $T_3$; the resulting answer undercounts the AIB players' games and distorts the average.

\method is the only method that identifies the complete synthesis. The extracted intent graph explicitly captures the hierarchical constraints (the value ``AIB'' and the attribute ``points'') together with the cross-tree join between clubs and players. Guided by this intent, the anchor-reachability initialization (Sec.~\ref{sec:online}) pinpoints $T_1$, $T_3$, and $T_4$ as anchor witnesses; the admissibility-filtered beam search then discovers $T_2$ as the bridge that realizes the cross-tree edge; and the covariance term of \metric recognizes $T_3$ and $T_4$ as complementary rather than redundant, retaining both shards. The final synthesis $T_1 \bowtie T_2 \bowtie (T_3 \cup T_4)$ resolves the structural disconnect and provides the downstream agent with exactly the evidence needed to compute the average.

\section{Related Work}\label{sec: related work}

\noindent \textbf{Table Retrieval.}
The goal of table retrieval in data lakes is to identify datasets that satisfy a user's specific analytical needs from massive, heterogeneous repositories. Existing approaches fundamentally operate under two paradigms. (1) \emph{Keyword- and semantic-based methods} (e.g., Aurum~\citep{fernandez2018aurum}, LSH Ensemble~\citep{zhu2016lshensemble}, and Google Dataset Search~\citep{brickley2019google}) retrieve tables that textually match specified query keywords or exhibit high dense vector similarity, with learned table--text encoders such as TaPas~\citep{herzig2020tapas}, TaBERT~\citep{yin2020tabert}, and Solo~\citep{wang2023solo} extending the paradigm to natural-language queries. (2) \emph{Structure-based methods} focus on discovering datasets that are structurally compatible with a user-provided base table, encompassing both \emph{joinable} table search (e.g., Josie~\citep{zhu2019josie}, DeepJoin~\citep{dong2023deepjoin}) and \emph{unionable} table search (e.g., Starmie~\citep{fan2023semantics}, SANTOS~\citep{khatiwada2023santos}).
Unlike existing methods that return isolated tables based on rigid inputs (e.g., exact keywords or a seed base table), \method interprets open-ended natural-language queries by casting retrieval as graph matching between a query-derived intent graph and the data lake graph, and returns an interconnected set of tables together with the join and union operations required to integrate them into a single coherent answer.

\noindent \textbf{Retrieval-Augmented Generation (RAG).}
RAG has emerged as the dominant paradigm for enhancing large language models (LLMs) with external knowledge, mitigating hallucinations by grounding generation in retrieved evidence~\citep{lewis2020rag}. Foundational architectures like REALM~\citep{guu2020realm} and Dense Passage Retrieval (DPR)~\citep{karpukhin2020dpr} optimize the retrieval stage using dual-encoder architectures that map queries and unstructured text documents into a shared dense vector space. On the generation side, approaches like Fusion-in-Decoder (FiD)~\citep{izacard2021fid} improve the LLM's ability to jointly process and reason over multiple retrieved passages.
Traditional tabular RAG retrieves serialized tables independently via flat semantic similarity, which misses the structural ``bridges'' that link fragmented evidence and hands the downstream agent an unjoinable context. \method instead realizes a structured RAG paradigm for tabular data lakes: it retrieves a logically interconnected evidence subgraph by jointly optimizing semantic relevance, structural compatibility, and evidence diversity in the single log-determinant objective \metric, so that the retrieved context is integrable by construction.

\section{Conclusion}
We formalized \problem as a graph-matching problem between a query-derived intent graph and a heterogeneous data lake graph, and proposed \method to solve it. \method couples two contributions: an information-theoretic reward (\metric) that unifies semantic relevance, structural compatibility, and evidence diversity in a single log-determinant objective; and an online search procedure that combines offline implicit Q-learning over the synthesis MDP with a three-stage pruning pipeline, reducing the candidate space from $\mathcal{O}(M \cdot N^{L-1})$ to $\mathcal{O}(BL)$. On Spider and BIRD, \method achieves the best Recall, Precision, F1, and Sufficiency among point-wise, greedy-expansion, and structure-aware baselines, with latency on par with greedy-expansion methods and substantially below the structure-aware MIP baseline. Natural extensions include jointly training $Q_\omega$ with a downstream executor for end-to-end synthesis, and lifting the formalism to multi-modal data lakes.

\bibliographystyle{ACM-Reference-Format}
\bibliography{reference}


\appendix

\section{Intent-Graph Extraction Prompt}
\label{app:prompt}

This appendix reproduces the full intent-graph extraction prompt issued to the LLM (GPT-5-mini in our default configuration), as described in Sec.~\ref{sec:graphs-problem}: the task description, the JSON output schema, the extraction rules, and the three in-context examples. The \texttt{\{query\}} placeholder is replaced by the user's natural-language query at inference time.

\begin{lstlisting}[basicstyle=\scriptsize\ttfamily,breaklines=true,frame=single,framerule=0.5pt,columns=fullflexible,xleftmargin=0pt]
You are an expert at decomposing analytical questions into a structured intent graph over a tabular data lake.

Given a natural-language QUERY, output a JSON object representing the intent graph as a forest of intent trees plus cross-tree join relations.

# OUTPUT SCHEMA

{
  "trees": [
    {
      "tree_id": "<unique short id>",
      "entity": "<target entity, e.g., 'Player', 'Game'>",
      "attributes": [
        {
          "attribute_id": "<unique short id>",
          "name": "<attribute, e.g., 'name', 'score', 'year'>",
          "values": ["<optional value constraints>"]
        }
      ]
    }
  ],
  "joins": [
    {
      "from_tree": "<tree_id>",
      "from_attribute": "<attribute_id in from_tree>",
      "to_tree": "<tree_id>",
      "to_attribute": "<attribute_id in to_tree>"
    }
  ]
}

# RULES

1. Use one tree per coherent data need: a single entity type together with its required attributes.
2. Each logical entity appears in at most one tree.
3. Add a join when two data needs are linked through a shared key, even if the link entity is not explicitly named in the query (introduce a
   bridge tree for it).
4. Include only entities, attributes, and values mentioned in or directly implied by the query. Do not invent unrelated attributes.
5. Do not include union relations. A single intent tree may be satisfied by multiple unioned tables at retrieval time; this is handled outside the intent graph and requires no graph annotation.
6. The "values" field is a (possibly empty) list of explicit value filters from the query. Leave it empty if the attribute is requested without a constraint.
7. Output valid JSON only. Do not include explanations, markdown fences, or any other text outside the JSON object.

# EXAMPLES

## Example 1 (single tree, value constraint).

QUERY: "List the top 5 universities in California by enrollment."

OUTPUT:
{
  "trees": [
    {
      "tree_id": "t1",
      "entity": "University",
      "attributes": [
        {"attribute_id": "a1", "name": "name", "values": []},
        {"attribute_id": "a2", "name": "state", "values": ["California"]},
        {"attribute_id": "a3", "name": "enrollment", "values": []}
      ]
    }
  ],
  "joins": []
}

## Example 2 (two trees, single join).

QUERY: "What is the average score of female players in basketball games during 2020?"

OUTPUT:
{
  "trees": [
    {
      "tree_id": "t1",
      "entity": "Player",
      "attributes": [
        {"attribute_id": "a1", "name": "id", "values": []},
        {"attribute_id": "a2", "name": "gender", "values": ["female"]}
      ]
    },
    {
      "tree_id": "t2",
      "entity": "Game",
      "attributes": [
        {"attribute_id": "a3", "name": "player_id", "values": []},
        {"attribute_id": "a4", "name": "score", "values": []},
        {"attribute_id": "a5", "name": "year", "values": ["2020"]},
        {"attribute_id": "a6", "name": "sport", "values": ["basketball"]}
      ]
    }
  ],
  "joins": [
    {"from_tree": "t1", "from_attribute": "a1",
     "to_tree": "t2", "to_attribute": "a3"}
  ]
}

## Example 3 (three trees, two joins, bridge entity).

QUERY: "What is the average points of AIB club players across matches?"

OUTPUT:
{
  "trees": [
    {
      "tree_id": "t1",
      "entity": "Club",
      "attributes": [
        {"attribute_id": "a1", "name": "id", "values": []},
        {"attribute_id": "a2", "name": "name", "values": ["AIB"]}
      ]
    },
    {
      "tree_id": "t2",
      "entity": "Player",
      "attributes": [
        {"attribute_id": "a3", "name": "id", "values": []},
        {"attribute_id": "a4", "name": "club_id", "values": []}
      ]
    },
    {
      "tree_id": "t3",
      "entity": "Match",
      "attributes": [
        {"attribute_id": "a5", "name": "player_id", "values": []},
        {"attribute_id": "a6", "name": "points", "values": []}
      ]
    }
  ],
  "joins": [
    {"from_tree": "t1", "from_attribute": "a1",
     "to_tree": "t2", "to_attribute": "a4"},
    {"from_tree": "t2", "from_attribute": "a3",
     "to_tree": "t3", "to_attribute": "a5"}
  ]
}

# QUERY

{query}

# OUTPUT
\end{lstlisting}

\end{document}